\documentclass[twocolumn,
showpacs,preprintnumbers,amsmath,amssymb]{revtex4}

\usepackage{amssymb,amsmath,amsthm}
\usepackage{epsfig}
\usepackage{multirow}
\newtheorem{Thm}{Theorem}

\newtheorem{Lem}[Thm]{Lemma}
\newtheorem{Prop}{Proposition}

\theoremstyle{definition}

\newcommand{\bra}[1]{{\left\langle #1 \right|}}
\newcommand{\ket}[1]{{\left| #1 \right\rangle}}

\newcommand{\T}{\mbox{$\mathrm{tr}$}}

\begin{document}
\title{Strong monogamy of quantum entanglement for multi-qubit W-class states}

\author{Jeong San Kim}
\email{freddie1@suwon.ac.kr} \affiliation{
 Department of Mathematics, University of Suwon, Kyungki-do 445-743, Korea
}
\date{\today}

\begin{abstract}
We provide a strong evidence for strong monogamy inequality of multi-qubit entanglement recently proposed in
[B.  Regula {\em et al.}, Phys. Rev. Lett. {\bf 113}, 110501 (2014)].
We consider a large class of multi-qubit generalized W-class states, and analytically show that the strong monogamy inequality
of multi-qubit entanglement is saturated by this class of states.
\end{abstract}

\maketitle

\section{Introduction}
\label{Intro}

Whereas classical correlation can be freely shared among parties in
multi-party systems, quantum entanglement is restricted in its
shareability; if a pair of parties are maximally entangled in
multipartite systems, they cannot have any entanglement~\cite{CKW, OV}
nor classical correlations~\cite{KW} with the rest of the system.
This restriction of entanglement shareability among multi-party systems is known as
the {\em monogamy of entanglement}~(MoE)~\cite{T04}.

MoE is at the heart of many quantum information and communication
protocols. For example, MoE is a key ingredient to make quantum
cryptography secure because it quantifies how much information
an eavesdropper could potentially obtain about the secret key to be
extracted~\cite{rg}. MoE also plays an important role
in condensed-matter physics such as
the frustration effects observed in Heisenberg antiferromagnets and
the $N$-representability problem for
fermions~\cite{anti}.

The first mathematical characterization of MoE was established by
Coffman-Kundu-Wootters(CKW) for three-qubit systems~\cite{CKW} as an inequality; for a three-qubit
pure state $\ket{\psi}_{ABC}$
with its one-qubit and two-qubit reduced density matrices
$\rho_A=\T_{BC}\ket{\psi}_{ABC}\bra{\psi}$,
$\T_C\ket{\psi}_{ABC}\bra{\psi}=\rho_{AB}$ and
$\T_B\ket{\psi}_{ABC}\bra{\psi}=\rho_{AC}$ respectively,
\begin{equation}
\tau\left(\ket{\psi}_{A|BC}\right) \geq \tau\left(\rho_{A|B}\right)+\tau\left(\rho_{A|C}\right),
\label{eq: CKW}
\end{equation}
where  $\tau\left(\ket{\psi}_{A|BC}\right)$
is the {\em tangle} of the pure state $\ket{\psi}_{ABC}$ quantifying the bipartite entanglement between $A$ and $BC$,
and $\tau\left(\rho_{A|B}\right)$ and $\tau\left(\rho_{A|C}\right)$ are the tangles
of the two-qubit reduced states $\rho_{AB}=\T_{C}\ket{\psi}_{ABC}\bra{\psi}$ and $\rho_{AC}=\T_{B}\ket{\psi}_{ABC}\bra{\psi}$,
respectively.

CKW inequality in~(\ref{eq: CKW}) shows the mutually exclusive nature of
three-qubit quantum entanglement in a quantitative way; more
entanglement shared between two qubits ($\tau\left(\rho_{A|B}\right)$) necessarily
implies less entanglement between the other two qubits ($\tau\left(\rho_{A|C}\right)$)
so that the summation does not exceed the total entanglement ($\tau\left(\ket{\psi}_{A|BC}\right)$).
Moreover, the residual entanglement from the difference between left and right-hand sides of Inequality~(\ref{eq: CKW})
is also interpreted as the genuine three-party entanglement, {\em three tangle} of $\ket{\psi}_{ABC}$
\begin{equation}
\tau\left(\ket{\psi}_{A|B|C}\right)=\tau\left(\ket{\psi}_{A|BC}\right)-\tau\left(\rho_{A|B}\right)-\tau\left(\rho_{A|C}\right),
\label{3tangle}
\end{equation}
which is invariant under the permutation of subsystems $A$, $B$ and $C$.
In this sense, $\tau\left(\ket{\psi}_{A|BC}\right)$ and $\tau\left(\rho_{A|B}\right)$ are also referred to
as the one tangle and two tangle, respectively~\cite{12tangle}.

Later, CKW inequality was generalized for multi-qubit systems~\cite{OV} and some cases of higher-dimensional
quantum systems in terms of various entanglement measures~\cite{KDS, KSRenyi, KT, KSU}. For general monogamy
inequality of multi-party entanglement, it was shown that
squashed entanglement~\cite{CW04} is a faithful entanglement measure
showing MoE of arbitrary quantum systems~\cite{BCY10}.

Recently, the three-tangle in Eq.~(\ref{3tangle})
was systematically generalized for arbitrary $n$-qubit quantum states, namely residual {\em $n$-tangle}~\cite{LA}.
Based on this generalization, the concept of {\em strong monogamy}(SM) inequality of multi-qubit entanglement
was proposed by conjecturing the nonnegativity of the  $n$-tangle~\cite{LA}. For the validity of SM inequality, an extensive numerical evidence was presented
for four qubit systems together with analytical proof for some cases of multi-qubit systems.
However, proving SM conjecture analytically for arbitrary multi-qubit states seems to be a formidable challenge due to the numerous optimization
processes arising in the definition of $n$-tangle.

Here we provide a strong evidence for SM inequality of multi-qubit entanglement; we consider a
large class of multi-qubit states, {\em generalized W-class states}, and analytically show that SM inequality proposed in~\cite{LA}
is saturated by this class of states. Because multi-qubit CKW inequality is known to be saturated by this generalized W-class states~\cite{Kim08},
this class of states are good candidates as possible counterexamples
for stronger version of monogamy inequalities.

The paper is organized as follows. In Sec.~\ref{Sec: Wstate}, we review the definition of generalized W-class states
for multi-qubit systems and provide some useful properties of this class in accordance with CKW inequality.
In Sec.~\ref{Subsec: strong}, we recall the concept of $n$-tangle as well as SM inequality of multi-qubit entanglement,
and show that multi-qubit SM inequality is saturated by generalized W-class states in Sec.~\ref{Subsec: SM W}.
In Sec.~\ref{Sec: Conclusion},  we summarize our results.

\section{Multi-qubit CKW inequality and the generalized W-class states}
\label{Sec: Wstate}

For $n$-qubit systems $\mathcal H_1 \otimes \cdots \otimes \mathcal H_n$ where $\mathcal H_j \cong \mathbb{C}^2$ for $j=1,\ldots,n$
and any $n$-qubit state $\ket{\psi}_{A_1 A_2 ... A_n} \in \mathcal H_1 \otimes \cdots \otimes \mathcal H_n$,
the three-qubit CKW inequality in (\ref{eq: CKW}) can be generalized as~\cite{OV}
\begin{equation}
\tau\left(\ket{\psi}_{A_1|A_2\cdots A_n}\right) \geq \sum_{j=2}^{n}\tau\left(\rho_{A_1|A_j}\right),
\label{eq: OV}
\end{equation}
where $\tau\left(\ket{\psi}_{A_1|A_2\cdots A_n}\right)$ is the tangle (or one tangle) of the pure state $\ket{\psi}_{A_1A_2\cdots A_n}$
with respect to the bipartition between $A_1$ and the other qubits $A_2\cdots A_n$
\begin{equation}
\tau\left(\ket{\psi}_{A_1|A_2\cdots A_n}\right)=4\det \rho_A,
\label{1tangle}
\end{equation}
and $\tau\left(\rho_{A_1|A_j}\right)$ is the tangle (or two tangle) of the two-qubit reduced density matrix $\rho_{A_1A_j}$
defined by convex-roof extension
\begin{equation}
\tau\left(\rho_{A_1|A_j}\right)=\bigg[\min\sum_h p_h \sqrt{\tau(\ket{\psi_h}_{A_1A_j})}\bigg]^2,
\label{2tangle}
\end{equation}
with the minimization taken over all possible pure state decompositions
\begin{equation}
\rho_{A_1A_j}=\sum_{h}p_{h}\ket{\psi_h}_{A_1A_j}\bra{\psi_h},
\label{decomp}
\end{equation}
for each $j=2,\cdots ,n$.

The $n$-qubit the generalized W-class state is defined as
\begin{align}
\ket{\psi}_{A_1 A_2 ... A_n}=&a\ket{00\cdots0}+b_1 \ket{10\cdots0}+b_2 \ket{01\cdots0}\nonumber\\
&+...+b_n \ket{00\cdots1}
\label{supWV}
\end{align}
with $|a|^2+\sum_{i=1}^{n}|b_j|^2 =1$~\cite{Kim08, GW}.
The term ``{\em generalized}'' naturally arises because Eq.~(\ref{supWV}) includes $n$-qubit W states as a special case
when $a=0$ and $b_j=1/\sqrt{n}$ for all $j$.

Before we further investigate strongly monogamous property of entanglement for the generalized W-class state in Eq.~(\ref{supWV}),
we recall a very useful property of quantum states proposed by Hughston-Jozsa-Wootters (HJW), which shows the unitary freedom
in the ensemble for density matrices~\cite{HJW}.
\begin{Prop} (HJW theorem)
The sets $\{|\tilde{\phi_i}\rangle\}$ and $\{|\tilde{\psi_j}\rangle\}$ of (possibly unnormalized) states generate the same density matrix
if and only if
\begin{equation}
|\tilde{\phi_i}\rangle=\sum_j u_{ij}|\tilde{\psi_j}\rangle\
\label{HJWeq}
\end{equation}
where $(u_{ij})$ is a unitary matrix of complex numbers, with indices $i$ and $j$, and we
{\em pad} whichever set of states $\{|\tilde{\phi_i}\rangle\}$ or $\{|\tilde{\psi_j}\rangle\}$ is smaller with additional zero vectors
so that the two sets have the same number of elements.
\label{HJWthm}
\end{Prop}

A direct consequence of Proposition~\ref{HJWthm} is the following; for two pure-state decompositions
$\sum_{i}p_{i}\ket{\phi_i}\bra{\phi_i}$ and $\sum_{j}q_{j}\ket{\psi_j}\bra{\psi_j}$,
they represent the same density matrix, that is $\rho=\sum_{i}p_{i}\ket{\phi_i}\bra{\phi_i}=\sum_{j}q_{j}\ket{\psi_j}\bra{\psi_j}$
if and only if $\sqrt{p_{i}}\ket{\phi_i}=\sum_{j}u_{ij}\sqrt{q_{j}}\ket{\psi_j}$ for some unitary matrix $u_{ij}$.
Now we have the following lemma showing that multi-qubit monogamy inequality in terms of one and two tangles in (\ref{eq: OV}) is
saturated by the generalized W-class states in (\ref{supWV}).
\begin{Lem}
For any multi-qubit system, multi-qubit CKW inequality is saturated by generalized W-class states,
that is,
\begin{align}
\tau\left(\ket{\psi}_{A_1|A_2\cdots A_n}\right) = \sum_{j=2}^{n}\tau\left(\rho_{A_1|A_j}\right),
\label{satWV}
\end{align}
for any $n$-qubit state $\ket{\psi}_{A_1|A_2\cdots A_n}$ in Eq.~(\ref{supWV})
\label{Lem: satWV}
\end{Lem}

\begin{proof}
Let us first consider the one tangle of $\ket{\psi}_{A_1 A_2 ... A_n}$ with respect to the bipartition between $A_1$ and the other qubits.
The reduced density matrix $\rho_{A_1}$ of subsystem $A_1$ is
\begin{align}
\rho_{A_1}=&\T_{A_2\cdots A_n}\ket{\psi}_{A_1 A_2 ... A_n}\bra{\psi}\nonumber\\
=&\left(a\ket{0}+b_1\ket{1}\right)_{A_1}\left(a^*\bra{0}+{b_1}^*\bra{1}\right)
+\sum_{j=2}^{n}|b_j|^2\ket{0}_{A_1}\bra{0},
\label{rho_A_1}
\end{align}
thus
\begin{align}
\tau\left(\ket{\psi}_{A_1|A_2\cdots A_n}\right)=4\det\rho_{A_1}=4|b_1|^2\sum_{j=2}^{n}|b_j|^2.
\label{onet}
\end{align}

For each $j=2,3,\cdots , n$, the reduced density matrix $\rho_{A_1A_j}$ of two-qubit subsystem $A_1A_j$ is
\begin{widetext}
\begin{align}
\rho_{A_1A_j}=&\T_{A_2\cdots \widehat{A_j} \cdots A_n}\ket{\psi}_{A_1 A_2 ... A_n}\bra{\psi}\nonumber\\
=&\left(a\ket{00}+b_1\ket{10}+b_j\ket{01}\right)_{A_1A_j}\left(a^*\bra{00}+b_1^*\bra{10}+b_j^*\bra{01}\right)
+\sum_{k\neq j}|b_k|^2\ket{00}_{A_1A_j}\bra{00},
\label{rho1i}
\end{align}
\end{widetext}
where $A_2\cdots \widehat{A}_j\cdots A_n = A_2\cdots A_{j-1}A_{j+1}\cdots A_n$ for each $j=2,3,\cdots, n$.
Here, we consider two-qubit (possibly) unnormalized states
\begin{align}
\ket{\tilde{x}}_{A_1A_j}=&a\ket{00}_{A_1A_j}+b_1\ket{10}_{A_1A_j}+b_j\ket{01}_{A_1A_j}\nonumber\\
\ket{\tilde{y}}_{A_1A_j}=&\sqrt{\sum_{k\neq j}|b_k|^2}\ket{00}_{A_1A_j},
\label{HJW}
\end{align}
which represents $\rho_{A_1A_j}$ as
\begin{equation}
\rho_{A_1A_j}=\ket{\tilde{x}}_{A_1A_j}\bra{\tilde{x}}+\ket{\tilde{y}}_{A_1A_j}\bra{\tilde{y}}.
\label{rho1irep}
\end{equation}

From the HJW theorem in Proposition~\ref{HJWthm}, we note that for any pure state decomposition of
\begin{equation}
\rho_{A_1 A_j}=\sum_{h=1}^{r}|\tilde{\phi_h}\rangle_{A_1 A_j} \langle\tilde{\phi_h}|,
\label{decomp}
\end{equation}
where
$|\tilde{\phi_h}\rangle_{A_1 A_j}$ is an unnormalized state in two-qubit subsystem $A_1A_j$,
there exists an $r\times r$ unitary matrix $(u_{hl})$ such that
\begin{equation}
|\tilde{\phi_h}\rangle_{A_1A_j}=u_{h1}\ket{\tilde{x}}_{A_1 A_j}+u_{h2}\ket{\tilde{y}}_{A_1 A_j},
\label{HJWrelation}
\end{equation}
for each $h$.

By considering the normalization $\ket{\phi_h}_{A_1 A_j}=|\tilde{\phi}_h\rangle_{A_1 A_j}/\sqrt{p_h}$
with $ p_h =|\langle\tilde{\phi}_h|\tilde{\phi}_h\rangle|$, we have the tangle of each two-qubit pure state
$\ket{\phi_h}_{A_1 A_j}$ as
\begin{align}
\tau\left(\ket{\phi_h}_{A_1|A_j}\right)=4\det\rho_{A_1}^{h}=\frac{4}{p_h^2}|u_{hj}|^4|b_1|^2|b_j|^2,
\label{tauphi_h}
\end{align}
where $\rho_{A_1}^{h}=\T_{A_i}\ket{\phi_h}_{A_1A_i}\bra{\phi_h}$ is the reduced density matrix of $\ket{\phi_h}_{A_1A_j}$ on subsystem $A_1$ for each $h$.
Moreover, the definition of two-tangle in Eq.~(\ref{2tangle}) together with Eq.~(\ref{tauphi_h}) lead us to
\begin{align}
\tau\left(\rho_{A_1A_j}\right)=&\bigg[\min_{\{p_h, \ket{\phi_h}\}}\sum_h p_h \sqrt{\tau\left(\ket{\phi_h}_{A_1|A_j}\right)}\bigg]^2\nonumber\\
=&\bigg[\min_{\{p_h, \ket{\phi_h}\}}\sum_h 2|u_{hj}|^2|b_1||b_j|\bigg]^2,\nonumber\\
=&4|b_1|^2|b_j|^2,
\label{2ti}
\end{align}
for each $j=2,\cdots ,n$.

Now Eqs.~(\ref{onet}) and (\ref{2ti}) implies Eq.~(\ref{satWV}), which completes the proof.
\end{proof}

For two tangle of two-qubit mixed state $\rho_{A_1A_j}$ in Eq.~(\ref{rho1i}), we need to deal with the minimization arising in the definition Eq.~(\ref{2tangle}).
In fact, any two-qubit mixed state can have an analytic entanglement measure called {\em concurrence}~\cite{WW}, whose analytic
evaluation can also be adapted for that of two tangle. However, the proof of Lemma~\ref{Lem: satWV} efficiently resolves this optimization problem by
considering all possible pure-state decompositions of $\rho_{A_1 A_j}$, which also shows a nice property of generalized W-class states;
the tangle of two-qubit reduced density matrix obtained from a generalized W-class state does not depend on
the choice of pure-state decomposition, $\rho_{A_1 A_j}=\sum_{h}{ p_h \ket{\phi_h}_{A_1 A_j}\bra{\phi_h}}$.
The following simple lemma shows another useful property about the structure of generalized W-class.
\begin{Lem}
Let $\ket{\psi}_{A_1\cdots A_n}$ be a generalized W-class state in Eq.~(\ref{supWV}).
For any $m$-qubit subsystems $A_1A_{j_1}\cdots A_{j_{m-1}}$ of $A_1\cdots A_n$ with $2 \leq m \leq  n-1$,
the reduced density matrix $\rho_{A_1A_{j_1}\cdots A_{j_{m-1}}}$ of $\ket{\psi}_{A_1\cdots A_n}$ is a mixture of a $m$-qubit generalized W-class state
and vacuum.
\label{reduced}
\end{Lem}
\begin{proof}
By a straightforward calculation, we obtain
\begin{align}
\rho_{A_1A_{j_1}\cdots A_{j_{m-1}}}=&\ket{\tilde{x}}_{A_1A_{j_1}\cdots A_{j_{m-1}}}\bra{\tilde{x}}\nonumber\\
&+\ket{\tilde{y}}_{A_1A_{j_1}\cdots A_{j_{m-1}}}\bra{\tilde{y}},
\label{mrho}
\end{align}
where
\begin{widetext}
\begin{align}
\ket{\tilde{x}}_{A_1A_{j_1}\cdots A_{j_{m-1}}}=
&\left(a\ket{00\cdots0}+b_1\ket{10\cdots0}+b_{j_1}\ket{01\cdots0}+
\cdots+b_{j_{m-1}}\ket{00\cdots1}\right)_{A_1A_{j_1}\cdots A_{j_{m-1}}},\nonumber\\
\ket{\tilde{y}}_{A_1A_{j_1}\cdots A_{j_{m-1}}}=&\sqrt{\sum_{k\in \{j_1, j_2, \cdots,j_{m-1} \}}|b_k|^2}\ket{00\cdots 0}_{A_1A_{j_1}\cdots A_{j_{m-1}}}
\label{xym}
\end{align}
\end{widetext}
are the unnormalized states in $m$-qubit subsystems $A_1A_{j_1}\cdots A_{j_{m-1}}$.
By considering the normalized states $\ket{x}_{A_1A_{j_1}\cdots A_{j_{m-1}}}=\ket{\tilde{x}}_{A_1A_{j_1}\cdots A_{j_{m-1}}}/\sqrt{p}$ with
$p=\langle\tilde{x}|\tilde{x}\rangle$ and $\ket{y}_{A_1A_{j_1}\cdots A_{j_{m-1}}}=\ket{\tilde{y}}_{A_1A_{j_1}\cdots A_{j_{m-1}}}/\sqrt{q}$ with
$q=\langle\tilde{y}|\tilde{y}\rangle$, we note that
\begin{align}
\rho_{A_1A_{j_1}\cdots A_{j_{m-1}}}=&p\ket{x}_{A_1A_{j_1}\cdots A_{j_{m-1}}}\bra{x}\nonumber\\
&+q\ket{y}_{A_1A_{j_1}\cdots A_{j_{m-1}}}\bra{y},
\end{align}
where $\ket{x}$ is a generalized W-class state and $\ket{y}$ is the vacuum, which completes the proof.
\end{proof}

\section{Strong monogamy inequality for multi-qubit generalized W-class states}
\label{Sec: SM W}
\subsection{Strong monogamy of multi-qubit entanglement}
\label{Subsec: strong}

The definition of three tangle in Eq.~(\ref{3tangle}) was generalized for arbitrary $n$-qubit quantum states~\cite{LA};
for an $n$-qubit pure state $\ket{\psi}_{A_1A_2\cdots A_n}$,
its {\em $n$-tangle} is defined as
\begin{align}
\tau\left(\ket{\psi}_{A_1|A_2|\cdots |A_n}\right)
=&\tau\left(\ket{\psi}_{A_1|A_2\cdots A_n}\right)\nonumber\\
&-\sum_{m=2}^{n-1} \sum_{\vec{j}^m}\tau\left(\rho_{A_1|A_{j^m_1}|\cdots |A_{j^m_{m-1}}}\right)^{m/2},
\label{eq:ntanglepure}
\end{align}
where the index vector $\vec{j}^m=(j^m_1,\ldots,j^m_{m-1})$ spans all the ordered subsets of the index set $\{2,\ldots,n\}$ with $(m-1)$ distinct elements.
For each $m=2,\cdots, n-1$, the $m$-tangle for multi-qubit mixed state is defined by convex-roof extension,
\begin{widetext}
\begin{equation}
\tau\left(\rho_{A_1|A_{j^m_1}|\cdots |A_{j^m_{m-1}}}\right)=\bigg[\min_{\{p_h, \ket{\psi_h}\}}\sum_h p_h
\sqrt{\tau\left(\ket{\psi_h}_{A_1|A_{j^m_1}|\cdots |A_{j^m_{m-1}}}\right)}\bigg]^2,
\label{ntanglemix}
\end{equation}
\end{widetext}
where the minimization of over all possible pure state decompositions
\begin{equation}
\rho_{A_1A_{j^m_1}\cdots A_{j^m_{m-1}}}=\sum_{h}p_{h}\ket{\psi_h}_{A_1A_{j^m_1}\cdots A_{j^m_{m-1}}}\bra{\psi_h}.
\label{decomp}
\end{equation}
Eq.~(\ref{eq:ntanglepure}) is an recurrent definition, that is, all the $m$ tangles $\tau\left(\rho_{A_1|A_{j^m_1}|\cdots |A_{j^m_{m-1}}}\right)$
for $2 \leq m \leq n-1$ need to appear to define the $n$ tangle $\tau\left(\ket{\psi}_{A_1|A_2|\cdots |A_n}\right)$.
We further note that Eq.~(\ref{eq:ntanglepure}) reduces to the two and three tangles when $n=2$ and $n=3$ respectively.
Based on this generalization, strong monogamy of multi-qubit entanglement was proposed
by conjecturing the nonnegativity of $n$-tangle Eq.~(\ref{eq:ntanglepure}),
\begin{align}
\tau\left(\ket{\psi}_{A_1|A_2\cdots A_n}\right)\geq\sum_{m=2}^{n-1} \sum_{\vec{j}^m}\tau\left(\rho_{A_1|A_{j^m_1}|\cdots |A_{j^m_{m-1}}}\right)^{m/2}.
\label{eq:SM}
\end{align}

The term {\em strong} naturally arises because
Inequality~(\ref{eq:SM}) is in fact {\em finer} than the $n$-qubit CKW inequality in (\ref{eq: OV})
\begin{align}
\tau\left(\ket{\psi}_{A_1|A_2\cdots A_n}\right)\geq&\sum_{j=2}^{n}\tau\left(\rho_{A_1|A_j}\right)\nonumber\\
&+\sum_{m=3}^{n-1} \sum_{\vec{j}^m}\tau\left(\rho_{A_1|A_{j^m_1}|\cdots |A_{j^m_{m-1}}}\right)^{m/2}\nonumber\\
\geq &\sum_{j=2}^{n}\tau\left(\rho_{A_1|A_j}\right).
\label{compar}
\end{align}
Moreover, Inequality~(\ref{eq:SM}) also encapsulates three-qubit CKW inequality in (\ref{eq: CKW}) for $n=3$,
thus Inequality~(\ref{eq:SM}) can be considered as another generalization of three-qubit CKW inequality in a stronger form.

\subsection{SM inequality for W-class states}
\label{Subsec: SM W}

For the validity of SM inequality in (\ref{eq:SM}), an extensive numerical evidence was presented
for four qubit systems together with analytical proof for some cases of multi-qubit systems.
However, providing an analytical proof of Inequality~(\ref{eq:SM}) for arbitrary multi-qubit states seems to be a formidable challenge
because there are numerous optimization processes arising in the recurrent definition of $n$-tangle~(\ref{eq:ntanglepure}).
Here we show that SM inequality holds for generalized W-class states in arbitrary multi-qubit systems.
Because Lemma~\ref{Lem: satWV} shows the multi-qubit CKW inequality is saturated by
generalized W-class states~\cite{Kim08}, this class of states are good candidates for possible violation of
stronger inequality, that is, SM inequality.

For the validity of SM inequality for generalized W-class states,
we first note that Inequality (\ref{eq:SM}) must be saturated by this class of states because of
Lemma~\ref{Lem: satWV} together with Inequalities (\ref{compar}). Thus we will show the residual term
\begin{align}
\sum_{m=3}^{n-1} \sum_{\vec{j}^m}\tau\left(\rho_{A_1|A_{j^m_1}|\cdots |A_{j^m_{m-1}}}\right)^{m/2}
\label{ktangleres}
\end{align}
in (\ref{compar}) is zero for any $n$-qubit generalized W-class state $\ket{\psi}_{A_1 A_2 ... A_n}$.
By using the mathematical induction on $m$, we further show that all the
$m$ tangles for $3\leq m \leq n-1$ is zero for generalized W-class states, that is,
\begin{align}
\tau\left(\rho_{A_1|A_{j^m_1}|\cdots |A_{j^m_{m-1}}}\right)=0,
\label{ktanglezero}
\end{align}
for all the index vectors $\vec{j}^m=(j^m_1,\ldots,j^m_{m-1})$ with $3\leq m \leq n-1$.

For $m=3$ and any index vector $\vec{j}=(j_1, j_2)$ with $j_1,~j_2 \in \{2,3,\cdots,n\}$, the left-hand side of Eq.~(\ref{ktanglezero}) becomes the three-tangle of the three-qubit
subsystem ${A_1A_{j_1}A_{j_2}}$~\cite{omit} where Lemma~\ref{reduced} leads us to the three-qubit reduced density matrix as
\begin{equation}
\rho_{A_1A_{j_1}A_{j_2}}=\ket{\tilde{x}}_{A_1A_{j_1}A_{j_2}}\bra{\tilde{x}}+\ket{\tilde{y}}_{A_1A_{j_1}A_{j_2}}\bra{\tilde{y}},
\label{rho123rep}
\end{equation}
with the three-qubit unnormalized states
\begin{align}
\ket{\tilde{x}}_{A_1A_{j_1}A_{j_2}}=a&\ket{000}_{A_1A_{j_1}A_{j_2}}+b_1\ket{100}_{A_1A_{j_1}A_{j_2}}\nonumber\\
&+b_{j_1}\ket{010}_{A_1A_{j_1}A_{j_2}}+b_{j_2}\ket{001}_{A_1A_{j_1}A_{j_2}}\nonumber\\
\ket{\tilde{y}}_{A_1A_{j_1}A_{j_2}}&=\sqrt{\sum_{k\neq j_1, j_2}|b_k|^2}\ket{000}_{A_1A_{j_1}A_{j_2}}.
\label{xy2}
\end{align}

The HJW theorem in Proposition~\ref{HJWthm} assures that for any pure state decomposition of $\rho_{A_1A_{j_1}A_{j_2}}$,
\begin{equation}
\rho_{A_1A_{j_1}A_{j_2}}=\sum_{h=1}^{r}|\tilde{\phi_h}\rangle_{A_1A_{j_1}A_{j_2}} \langle\tilde{\phi_h}|,
\label{decomp123}
\end{equation}
where
$|\tilde{\phi_h}\rangle_{A_1A_{j_1}A_{j_2}}$ is an unnormalized state in three-qubit subsystem ${A_1A_{j_1}A_{j_2}}$,
there exists an $r\times r$ unitary matrix $(u_{hl})$ that makes a relation between pure state ensembles of $\rho_{A_1A_{j_1}A_{j_2}}$  as
\begin{equation}
|\tilde{\phi_h}\rangle_{A_1A_{j_1}A_{j_2}}=u_{h1}\ket{\tilde{x}}_{A_1A_{j_1}A_{j_2}}+u_{h2}\ket{\tilde{y}}_{A_1A_{j_1}A_{j_2}}.
\label{phih123}
\end{equation}
Here we note, for each $h$, $|\tilde{\phi_h}\rangle_{A_1A_{j_1}A_{j_2}}$ in Eq.~(\ref{phih123}) is a (unnormalized) superposition of a three-qubit W-class state
and vacuum. Thus Lemma~\ref{Lem: satWV} assures that the normalized state $\ket{\phi_h}_{A_1A_{j_1}A_{j_2}}=|\tilde{\phi}_h\rangle_{A_1A_{j_1}A_{j_2}}/\sqrt{p_h}$
with $ p_h =|\langle\tilde{\phi}_h|\tilde{\phi}_h\rangle|$ satisfies Eq.~(\ref{satWV}), that is, the three-tangle of
$|\tilde{\phi_h}\rangle_{A_1A_{j_1}A_{j_2}}$ in Eq.~(\ref{eq:ntanglepure}) is zero,
\begin{align}
\tau\left(\ket{\phi_h}_{A_1|A_{j_1}|A_{j_2}}\right)=&\tau\left(\ket{\phi_h}_{A_1|A_{j_1}A_{j_2}}\right)\nonumber\\
&-\tau\left(\rho_{A_1|A_{j_1}}\right)
-\tau\left(\rho_{A_1|A_{j_2}}\right)\nonumber\\
=&0,
\label{phi123zero}
\end{align}
for each $h$.

Eq.~(\ref{phi123zero}) implies that three-qubit pure state that arises in any pure state ensemble of $\rho_{A_1A_{j_1}A_{j_2}}$ in Eq.~(\ref{decomp123})
has zero as its three tangle value. Thus, from the definition of $n$-tangle for multi-qubit mixed state in Eq.~(\ref{2tangle}), we have
\begin{align}
\tau\left(\rho_{A_1|A_{j_1}|A_{j_2}}\right)=&\bigg[\min_{\{p_h, \ket{\phi_h}\}}\sum_h p_h \sqrt{\tau\left(\ket{\phi_h}_{A_1|A_{j_1}|A_{j_2}}\right)}\bigg]^2\nonumber\\
=&0,
\label{taurho123}
\end{align}
for any the three-qubit reduced density matrix $\rho_{A_1A_{j_1}A_{j_2}}$ of $\ket{\psi}_{A_1 A_2 ... A_n}$.

We now assume the induction hypothesis for Eq.~(\ref{ktanglezero}); for any $(m-1)$-qubit reduced density matrix
$\rho_{A_1A_{j_1}A_{j_2}\cdots A_{j_{m-2}}}$ of the generalized W-class states in Eq.~(\ref{supWV}), we assume its $(m-1)$ tangle is zero,
\begin{align}
\tau\left(\rho_{A_1|A_{j_1}|A_{j_2}|\cdots|A_{j_{m-2}}}\right)=0,
\label{induct}
\end{align}
and show its validity for $m\leq n-1$.

For any index vector $\vec{j}=(j_1, j_2, \ldots, j_{m-1})$ with $\{j_1,~j_2, \ldots, j_{m-1}\}\subseteq\{2,3,\cdots,n\}$,
Lemma~\ref{reduced} assures that the $m$-qubit reduced density matrix of $\ket{\psi}_{A_1 A_2 ... A_n}$ on subsystems $A_1A_{j_1}\cdots A_{j_{m-1}}$ is
\begin{align}
\rho_{A_1A_{j_1}\cdots A_{j_{m-1}}}=&\ket{\tilde{x}}_{A_1A_{j_1}\cdots A_{j_{m-1}}}\bra{\tilde{x}}\nonumber\\
&+\ket{\tilde{y}}_{A_1A_{j_1}\cdots A_{j_{m-1}}}\bra{\tilde{y}},
\label{mrhom}
\end{align}
where
$\ket{\tilde{x}}_{A_1A_{j_1}\cdots A_{j_{m-1}}}$ and $\ket{\tilde{y}}_{A_1A_{j_1}\cdots A_{j_{m-1}}}$ are the $m$-qubit unnormalized states in Eq.~(\ref{xym}).
By HJW theorem in Proposition~\ref{HJWthm}, we note that any pure state decomposition
\begin{equation}
\rho_{A_1A_{j_1}\cdots A_{j_{m-1}}}=\sum_{h=1}^{r}|\tilde{\phi_h}\rangle_{A_1A_{j_1}\cdots A_{j_{m-1}}} \langle\tilde{\phi_h}|,
\label{decompn}
\end{equation}
is related with Eq.~(\ref{mrhom}) by some
$r\times r$ unitary matrix $(u_{hl})$ such that
\begin{align}
|\tilde{\phi_h}\rangle_{A_1A_{j_1}\cdots A_{j_{m-1}}}=&u_{h1}\ket{\tilde{x}}_{A_1A_{j_1}\cdots A_{j_{m-1}}}\nonumber\\
&+u_{h2}\ket{\tilde{y}}_{A_1A_{j_1}\cdots A_{j_{m-1}}},
\label{phihm}
\end{align}
for each $h$. Furthermore, the normalization $\ket{\phi_h}_{A_1A_{j_1}\cdots A_{j_{m-1}}}=|\tilde{\phi}_h\rangle_{A_1A_{j_1}\cdots A_{j_{m-1}}}/\sqrt{p_h}$
with $ p_h =|\langle\tilde{\phi}_h|\tilde{\phi}_h\rangle|$ is a superposition of a $m$-qubit generalized W-class state and vacuum, which is again a generalized
W-class state.

From the definition of pure state tangle in Eq.~(\ref{eq:ntanglepure}), the $m$ tangle of each $m$-qubit pure state
$\ket{\phi_h}_{A_1A_{j_1}\cdots A_{j_{m-1}}}$ is
\begin{widetext}
\begin{align}
\tau\left(\ket{\phi_h}_{A_1|A_{j_1}|\cdots|A_{j_{m-1}}}\right)
=&\tau\left(\ket{\phi_h}_{A_1|A_{j_1}\cdots A_{j_{m-1}}}\right)-\sum_{k=2}^{m-1} \sum_{\vec{i}^k}\tau\left(\rho^h_{A_1|A_{i_1}|\cdots |A_{i_{k-1}}}\right)^{k/2},
\label{mtanglepure1}
\end{align}
\end{widetext}
where $\rho^h_{A_1A_{i_1}\cdots A_{i_{k-1}}}$ is the reduced density matrix of $\ket{\phi_h}_{A_1A_{j_1}\cdots A_{j_{m-1}}}$ on $k$-qubit subsystems
${A_1A_{i_1}\cdots A_{i_{k-1}}}$ with the index vector $\vec{i}^k=(i_1, i_2, \cdots, i_{k-1})$ for $\{i_1,~i_2, \cdots, i_{k-1}\} \subseteq \{j_1, j_2,\cdots,j_{m-1}\}$.
Let us further divide the last term of the inequality into the summation of two tangles and the others;
\begin{widetext}
\begin{align}
\tau\left(\ket{\phi_h}_{A_1|A_{j_1}|\cdots|A_{j_{m-1}}}\right)
=&\tau\left(\ket{\phi_h}_{A_1|A_{j_1}\cdots A_{j_{m-1}}}\right)-\sum_{l=1}^{m-1} \tau\left(\rho^h_{A_1|A_{j_l}} \right) -\sum_{k=3}^{m-1} \sum_{\vec{i}^k}\tau\left(\rho^h_{A_1|A_{i_1}|\cdots |A_{i_{k-1}}}\right)^{k/2}.
\label{mtanglepure2}
\end{align}
\end{widetext}

For each $k=3, \cdots ,m-1$, $\rho_{A_1A_{i_1}\cdots A_{i_{k-1}}}$ in the last summation of Eq.~(\ref{mtanglepure2}) is a $k$-qubit reduced density matrix of
the generalized W-class state $\ket{\phi_h}_{A_1A_{j_1}\cdots A_{j_{m-1}}}$, therefore the induction hypothesis assures that its $k$ tangle is zero;
\begin{equation}
\tau\left(\rho^h_{A_1|A_{i_1}|\cdots |A_{i_{k-1}}}\right)=0,
\label{ktau0}
\end{equation}
for each $k=3,\cdots ,m-1$ and index vector $\vec{i}^k=(i_1, i_2, \cdots, i_{k-1})$.
Furthermore, Lemma~\ref{Lem: satWV} implies that the usual monogamy inequality in terms of one and two tangles is saturated by
$\ket{\phi_h}_{A_1A_{j_1}\cdots A_{j_{m-1}}}$;
\begin{align}
\tau\left(\ket{\phi_h}_{A_1|A_{j_1}\cdots A_{j_{m-1}}}\right)=\sum_{l=1}^{m-1}\tau\left(\rho^h_{A_1|A_{j_l}} \right),
\label{satphi}
\end{align}
for each $h$.

Eq.~(\ref{ktau0}) together with Eq.~(\ref{satphi}) imply that
\begin{align}
\tau\left(\ket{\phi_h}_{A_1|A_{j_1}|\cdots|A_{j_{m-1}}}\right)=0
\label{tphi0}
\end{align}
for each $\ket{\phi_h}_{A_1A_{j_1}\cdots A_{j_{m-1}}}$ that arises in the decomposition of $\rho_{A_1A_{j_1}\cdots A_{j_{m-1}}}$,
\begin{align}
\rho_{A_1A_{j_1}\cdots A_{j_{m-1}}}=&\sum_{h=1}^{r}|\tilde{\phi_h}\rangle_{A_1A_{j_1}\cdots A_{j_{m-1}}} \langle\tilde{\phi_h}|\nonumber\\
=&\sum_{h=1}^{r}p_h\ket{\phi_h}_{A_1A_{j_1}\cdots A_{j_{m-1}}}\bra{\phi_h}.
\label{decompn2}
\end{align}
Thus, from the definition of $n$-tangle for multi-qubit mixed state in Eq.~(\ref{2tangle}), we have
\begin{widetext}
\begin{align}
\tau\left(\rho_{A_1|A_{j_1}|\cdots|A_{j_{m-1}}} \right)=&\bigg[\min_{\{p_h, \ket{\phi_h}\}}\sum_h p_h \sqrt{\tau\left(\ket{\phi_h}_{A_1|A_{j_1}|\cdots|A_{j_{m-1}}}\right)}\bigg]^2
=0,
\label{taurhom}
\end{align}
\end{widetext}
for any the $m$-qubit reduced density matrix $\rho_{A_1A_{j_1}\cdots A_{j_{m-1}}}$ of $\ket{\psi}_{A_1 A_2 ... A_n}$ with $3\leq m \leq n-1$.
Now Eq.~(\ref{taurhom}) together with Lemma~\ref{Lem: satWV}, we have the following theorem showing the saturation of multi-qubit SM inequality
by generalized W-class states.
\begin{Thm}
The strong monogamy inequality of multi-qubit entanglement is saturated by the generalized W-class states;
\begin{align}
\tau\left(\ket{\psi}_{A_1|A_2\cdots A_n}\right)=\sum_{m=2}^{n-1} \sum_{\vec{j}^m}\tau\left(\rho_{A_1|A_{j^m_1}|\cdots |A_{j^m_{m-1}}}\right)^{m/2},
\label{eq:SMsat}
\end{align}
for any multi-qubit generalized W-class state in Eq.~(\ref{supWV}).
\label{thm: smono}
\end{Thm}
\section{Conclusions}\label{Sec: Conclusion}

We have considered a large class of multi-qubit generalized W-class states,
and provided a strong evidence for SM inequality of multi-qubit entanglement.
Although providing an analytical proof of SM inequality for arbitrary multi-qubit states seems to be a formidable challenge
because there are numerous optimization processes arising in the recurrent definition of $n$-tangle, we have successfully resolved
this problem by investing the structural properties of W-class states, and analytically shown that strong monogamy inequality is saturated
by this class of states.

Our result characterizes the strongly monogamous nature of arbitrary multi-qubit W-class states.
Noting the importance of the study on multipartite
entanglement, our result can provide a rich reference for future
work on the study of entanglement in complex quantum systems.

\section*{Acknowledgments}
The author would like to appreciate G. Adesso, S. Lee, S. D. Martino and B. Regula for helpful discussions.
This research was supported by Basic Science Research Program through the National Research Foundation of Korea(NRF)
funded by the Ministry of Education, Science and Technology(2012R1A1A1012246).


\begin{thebibliography}{1}
\bibitem{CKW}
V. Coffman, J. Kundu and W. K. Wootters,
Phys. Rev. A {\bf 61}, 052306 (2000).

\bibitem{OV}
T. Osborne and F. Verstraete,
Phys. Rev. Lett. {\bf 96}, 220503 (2006).

\bibitem{KW}
M. Koashi and A. Winter, Phys. Rev. A {\bf69}, 022309 (2004).


\bibitem{T04}
B. M. Terhal, IBM J. Research and Development {\bf 48}, 71 (2004).

\bibitem{rg}
J. M. Renes and M. Grassl,
Phys. Rev. A {\bf74}, 022317 (2006).

\bibitem{anti}
A. J. Coleman and V. I. Yukalov,
Lecture Notes in Chemistry Vol. {\bf 72} (Springer-Verlag, Berlin, 2000).

\bibitem{12tangle}
In~\cite{LA}, the notation $\tau_{A|B|C}^{(3)}\left(\ket{\psi}\right)$
is used where the superscript $(3)$ is to indicate three tangle. Here we use the notation $\tau\left(\ket{\psi}_{A|B|C}\right)$
instead of $\tau_{A|B|C}^{(3)}\left(\ket{\psi}\right)$ because it does not cause any confusion.

\bibitem{KDS}
J. S. Kim, A. Das and B. C. Sanders,
Phys. Rev. A {\bf 79}, 012329 (2009).

\bibitem{KSRenyi}
J. S. Kim and B. C. Sanders,
J. Phys. A: Math. and Theor. {\bf 43}, 445305 (2010).

\bibitem{KT}
J. S. Kim,
Phys. Rev. A. {\bf 81}, 062328 (2010).

\bibitem{KSU}
J. S. Kim and B. C. Sanders,
J. Phys. A: Math. and Theor. {\bf 44}, 295303 (2011).

\bibitem{CW04}
M. Christandl and A. Winter,
J. Math. Phys. {\bf 45}, p.~829--840 (2004).

\bibitem{BCY10}
F. G. S. L. Brandao, M. Christandl and J. Yard,
Commun. Math. Phys. {\bf 306}, 805 (2011).

\bibitem{LA}
B. Regula, S. D. Martino, S. Lee and G. Adesso,
Phys. Rev. Lett. {\bf 113}, 110501 (2014).

\bibitem{Kim08}
J. S. Kim and B. C. Sanders,
J. Phys. A {\bf 41}, 495301 (2008).

\bibitem{GW}
Generalized W-class state was defined in \cite{Kim08} as a superposition
of all $n$-qubit product states having $1$ in a single qubit and vacuum elsewhere, that is,
$\ket{\psi}_{A_1 A_2 ... A_n} =b_1 \ket{10\cdots0}+b_2 \ket{01\cdots0}+...+b_n \ket{00\cdots1}$.
Here, we consider a more general class of states in Eq.~(\ref{supWV}) including the above states as a special case
when $a=0$.

\bibitem{HJW} L.~P.~Hughston, R.~Jozsa and W.~K.~Wootters,
Phys. Lett. A {\bf 183}, 14 (1993).

\bibitem{WW}
W. K. Wootters,
Phys. Rev. Lett. {\bf 80}, 2245 (1998).

\bibitem{omit}
For simplicity, we omit the superscript $m$ in the components of the
index vector $\vec{j}^m=(j^m_1,\ldots,j^m_{m-1})$ because there is no confusion.
\end{thebibliography}
\end{document}